\documentclass[11pt,a4paper]{article}    

\usepackage[margin=1in]{geometry} 
\usepackage{amsmath,amssymb,amsthm}
\usepackage{color}
\usepackage{relsize}
\usepackage{cite}
\usepackage[none]{hyphenat}

\newtheorem{theorem}{Theorem}[section]

\newtheorem{lemma}[theorem]{Lemma}

\newtheorem{remark}{Remark}

\newcommand{\be}{\begin{equation}}
\newcommand{\ee}{\end{equation}}
\newcommand{\bea}{\begin{eqnarray}}
\newcommand{\eea}{\end{eqnarray}}

\usepackage[perpage,symbol*]{footmisc}

\usepackage{authblk}

\numberwithin{equation}{section}
\title{Gaussian Unitary Ensembles with Jump Discontinuities, \\PDEs and the Coupled Painlev\'{e} IV System}

\author[1]{Yang Chen}
\author[2,\thanks{Author to whom any correspondence should be addressed. E-mail: lvshulin1989@163.com}]{Shulin Lyu}
\affil[1]{Department of Mathematics, Faculty of Science and Technology, University of Macau, Macau, China}
\affil[2]{School of Mathematics and Statistics, Qilu University of Technology (Shandong Academy of Sciences), Jinan 250353, China}

\date{}
\begin{document}
\maketitle
\sloppy{
\begin{abstract}
We study the Hankel determinant generated by the Gaussian weight with jump discontinuities at $t_1,\cdots,t_m$. By making use of a pair of ladder operators satisfied by the associated monic orthogonal polynomials and three supplementary conditions, we show that the logarithmic derivative of the Hankel determinant satisfies a second order partial differential equation which is reduced to the $\sigma$-form of a Painlev\'{e} IV equation when $m=1$. Moreover, under the assumption that $t_k-t_1$ is fixed for $k=2,\cdots,m$,
by considering the Riemann-Hilbert problem for the orthogonal polynomials, we construct direct relationships between the auxiliary quantities introduced in the ladder operators and solutions of a coupled Painlev\'{e} IV system.
\end{abstract}

$\mathbf{Keywords}$:  Gaussian unitary ensembles; Hankel determinant; Orthogonal polynomials;

Painlev\'{e} equations

$\mathbf{Mathematics\:\: Subject\:\: Classification\:\: 2020}$: 15B52; 33E17; 42C05

\section{Introduction}
The $n$-dimensional Gaussian unitary ensemble (GUE for short) is a collection of $n\times n$ Hermitian random matrices whose eigenvalues have the following joint probability density function
\begin{align*}
p(x_1,\cdots,x_n)=\frac{1}{Z_n}\cdot\frac{1}{n!}\prod_{1\leq j<k\leq n}(x_j-x_k)^2\prod_{\ell=1}^n {\rm e}^{-x_{\ell}^2},
\end{align*}
where $x_{\ell}\in(-\infty,+\infty)$ for $\ell=1,\cdots,n$.
The reader may consult \cite[sections 2.5, 2.6 and 3.3]{Mehta} for more information. The normalization constant $n!Z_n$, also known as the partition function, is a Selberg type integral and equal to $\pi^{n/2}2^{-n(n-1)/2}\prod\limits_{j=2}^{n}j!$ (see, for instance \cite[equation (17.6.7)]{Mehta}).
For an arbitrary interval $I\subset(-\infty,+\infty)$, the probability that all the eigenvalues of GUE are contained in $I$ is given by
 \begin{align}
 \mathbb{P}_I:=&Prob\left(\{x_{\ell}\in I, \ell=1,\cdots,n\}\right)\nonumber\\
 =&\frac{1}{Z_n}\cdot\frac{1}{n!}\int_{I^n}\prod_{1\leq j<k\leq n}(x_j-x_k)^2\prod_{\ell=1}^n {\rm e}^{-x_{\ell}^2}dx_1\cdots dx_n\nonumber\\
 =&\frac{1}{Z_n}\det\left(\int_I x^{i+j}{\rm e}^{-x^2}dx\right)_{i,j=0}^{n-1},\label{probD}
 \end{align}
where the third equality is due to the following representation of the multiple integral as the determinant of a Hankel matrix, which results from Heine's formula (see e.g. \cite[sections 2.1 and 2.2]{Szego})
\begin{align}\label{Heine}
\frac{1}{n!}\int_{I^n}\prod_{1\leq j<k\leq n}(x_j-x_k)^2\prod_{\ell=1}^n w(x_{\ell})dx_1\cdots dx_n=\det\left(\int_I x^{i+j}w(x)dx\right)_{i,j=0}^{n-1}.
\end{align}
Here $w(x)$ is an arbitrary positive weight function whose moments of all orders exist.

In the present paper, we study the following Hankel determinant
\begin{align}
D_n\left(\vec{t}\,\right):=&\det\left(\int_{-\infty}^{+\infty} x^{i+j}w(x;\vec{t}\,)dx\right)_{i,j=0}^{n-1}\label{def-Dn}\\
=&\frac{1}{n!}\int_{(-\infty,+\infty)^n}\prod_{1\leq j<k\leq n}(x_j-x_k)^2\prod_{\ell=1}^n w(x_{\ell};\vec{t}\,)dx_1\cdots dx_n,\nonumber
\end{align}
where $\vec{t}=\left(t_1,\cdots,t_m\right)$ with $t_1<\cdots<t_m$ and the second identity comes from \eqref{Heine}. Here the weight function $w(x;\vec{t}\,)$ is the Gaussian weight multiplied by a factor with $m$ jump discontinuities, i.e.
\begin{align}\label{weight}
w(x;\vec{t}\,):={\rm e}^{-x^2}\biggl(\omega_0+\sum_{k=1}^m \omega_k\cdot\theta(x-t_k)\biggr),
\end{align}
where $\sum\limits_{k=0}^\ell \omega_k\geq0$ for $\ell=0,1,\cdots,m$,
and $\theta(x-t_k)$ is 1 for $x>t_k$ and 0 otherwise.
It is well known that $D_n(\vec{t}\,)$ admits an alternative representation (see e.g. \cite{Szego})
\begin{align}\label{Dprodh}
D_n(\vec{t}\,)=\prod_{k=0}^{n-1}h_k(\vec{t}\,),
\end{align}
where $h_k(\vec{t}\,)$ is the square of the $L^2$-norm of the monic polynomials orthogonal with respect to $w(x;\vec{t}\,)$, namely
\begin{align}\label{m-or}
\int_{-\infty}^{+\infty}P_j(x;\vec{t}\,)P_k(x;\vec{t}\,)w(x;\vec{t}\,)dx=h_{k}(\vec{t}\,)\delta_{jk},\qquad j,k\geq0.
\end{align}
Here $\delta_{jk}$ is 1 for $j=k$ and 0 otherwise, and
\begin{align}\label{defPn}
P_k\left(x;\vec{t}\,\right):=x^k+p(k,\vec{t}\,)x^{k-1}+\cdots+P_k(0;\vec{t}\,),\qquad k\geq1,
\end{align}
with $P_0\left(x;\vec{t}\,\right):=1$. We will see in the forthcoming discussions that the orthogonal polynomials play a vital role in the study of random matrices.

When the weight function \eqref{weight} has one jump discontinuity (i.e. $m=1$) or two jump discontinuities (i.e. $m=2$), we find by combining \eqref{probD} with \eqref{def-Dn} that the Hankel determinant is intimately related to the probability distribution of GUE. In fact, for $m=1$ with the two parameters $\omega_0=0, \omega_1=1$ or $\omega_0=1, \omega_1=-1$ in the weight function, the quantity $D_n(t_1)/Z_n$ represents the smallest eigenvalue distribution $\mathbb{P}_{[t_1,+\infty)}$ or the largest eigenvalue distribution $\mathbb{P}_{(-\infty,t_1]}$ respectively. They were investigated by Tracy and Widom via the operator theory, and the integral representations for the distribution functions were established in terms of solutions of a Painlev\'{e} IV equation ($P_{IV}$ for short) for finite $n$ \cite{TW163} and of a Painlev\'{e} II equation ($P_{II}$ for short) at the soft edge \cite{TW159}. When $m=2$ with $\omega_0=0,\omega_1=1,\omega_2=-1$ or $\omega_0=1,\omega_1=-1,\omega_2=1$, the quantity $D_n(t_1,t_2)/Z_n$ denotes the probability that all or no eigenvalues of GUE lie in $[t_1,t_2]$ respectively. The former case was studied in \cite{BasorChenZhang} by using the ladder operator approach and the logarithmic derivative of the distribution function was shown to satisfy a two-variable generalization of a $P_{IV}$. When $\sum_{k=0}^\ell \omega_k$ is strictly positive for all $\ell\in\{0,1,\cdots,m\}$, the Hankel determinant $D_n\left(\vec{t}\,\right)$ denotes the moment generating function of the counting statistics, up to a constant. In this case, the discontinuities are called Fisher-Hartwig singularities of jump type and the asymptotics of $D_n\left(\vec{t}\,\right)$ was derived in \cite{Charlier}.

Using the definition of $P_k(x;\vec{t}\,)$ given by \eqref{m-or}-\eqref{defPn}, one derives a three-term recurrence relation for $P_k(x;\vec{t}\,)$ and the Christoffel-Darboux formula (to be presented in the next section). With these properties, one shows that $\{P_k(x;\vec{t}\,)\}$ satisfy a lowering and a raising operators and three supplementary conditions $(S_1), (S_2), (S_2')$, which form the basis of the ladder operator approach to study Hankel determinants. From the pair of ladder operators, one obtains a second order differential equation for the monic orthogonal polynomials and the properties of the coefficients with different choices of the potential function $v(x)=-\ln w(x)$ were described in \cite{ChenIsmail}. The reader may also refer to \cite{Ismail}, \cite[section 3.2]{Szego} and \cite{Assche} for an in-depth description of the ladder operator approach.

We state here some literature on the application of the ladder operator approach to different problems in random matrix theory. Magnus \cite{Magnus} took this approach to study semi-classical orthonormal polynomials and derive Painlev\'{e} equations. The approach was also used to determine the recurrence coefficients of the classic (monic) Jacobi polynomials \cite{ChenIsmail04}. In addition, it was adopted to study the linear statistics of the Laguerre unitary ensemble \cite{ChenIts10} which arises from the integrable quantum field, and to characterize the outage capacity of multiple-input multiple-output wireless communication system \cite{ChenHaqMcKay13,ChenMcKay}. For more recent works, see e.g. \cite{HanChen}, \cite{MinChen21}, \cite{MinChen22-AMP}.

When the dimension of the unitary ensemble tends to $\infty$, by conducting the celebrated Deift-Zhou steepest descent analysis \cite{DeiftZhou} to the Riemann-Hilbert problem satisfied by the associated orthogonal polynomials (known as the RH method), researchers described the asymptotic behavior of the partition function \cite{BMM}, the orthogonal polynomials \cite{ChenChenFan19}, the correlation kernel \cite{ChenChenFan19-1} and the probability distribution function \cite{XuZhao20}.  See Deift \cite{Deift} for a full description of this method.

The coupled Painlev\'{e} systems are recently used to characterize random matrix ensembles that involve more than one variables. For example, the partition function and the correlation kernel of GUE with pole singularities near the soft edge were described by the coupled Painlev\'{e} XXXIV system \cite{DaiXuZhang19}, while the gap probability of GUE with pole singularities at the hard edge admitted an integral representation in terms of the solutions of a coupled Painlev\'{e} III system \cite{DaiXuZhang18}; the Fredholm determinants of multi-interval Airy and Bessel kernels were expressed by the solutions of coupled $P_{II}$ \cite{cd18} and $P_{IV}$ \cite{cd} systems respectively. See also \cite{ACM,BasorChenZhang,ChenHaqMcKay13, LyuGriffinChen}.

When the weight function \eqref{weight} has only one jump discontinuity (i.e. $m=1$), under suitable double scaling and by employing the RH method, the asymptotic behavior of the associated orthogonal polynomials was characterized by the solutions of a Painlev\'{e} XXXIV equation when the jump is located at the edge of the spectrum \cite{XuZhao11}. When the dimension $n$ of the Hankel determinant $D_n(t_1)$ is finite, via the ladder operator approach, it was shown in \cite{ChenP, MinChen19} that $D_n(t_1)$ was intimately related to the $\tau$-function of $P_{IV}$. The two-jump case was considered in \cite{MinChen19}. By taking the ladder operator approach, the logarithmic derivative of the Hankel determinant was shown to satisfy a second order partial differential equation (PDE for short) which can be regarded as a two-variable generalization of the $\sigma$-form of $P_{IV}$.

The Hankel determinant $D_n(\vec{t}\,)$ for generic $m$ was investigated in \cite{WuXu20} and we briefly describe what was accomplished. For finite $n$, by studying the RH problem for $\{P_k(x;\vec{t}\,)\}$ and with the aid of the Lax pair, the logarithmic derivative of $D_n(\vec{t}\,)$ turned out to be the Hamiltonian of a coupled $P_{IV}$ system.
When $n\rightarrow\infty$ and the jump discontinuities $\{t_k, k=1,\cdots,m\}$ go to the edge of the spectrum, by adopting the RH method,
the asymptotic expressions for $D_n(\vec{t}\,)$ and $\{P_k(x;\vec{t}\,)\}$ were established in terms of the solutions of a coupled $P_{II}$ system.
The coupled $P_{IV}$ and $P_{II}$ systems for $m=2$ were reproduced in \cite{LyuChen20} by making use of the finite-$n$ results of \cite{MinChen19} obtained via the ladder operator approach. This finding may provide a further insight into the connection between the ladder operators and the RH problem satisfied by the orthogonal polynomials.

In this paper, we will employ the ladder operator approach to establish second order PDEs to characterize $D_n(\vec{t}\,)$. Our analysis is not simple generalization of the $m=2$ case which was studied in \cite{MinChen19}: to obtain identities from $(S_2')$, we need a technical algebraic operation; to derive Riccati equations for the auxiliary quantities, we make use of the compatibility conditions obtained from the differentiation of the orthogonality relation for the orthogonal polynomials. The PDE satisfied by the logarithmic derivative of $D_n(\vec{t}\,)$ looks simple in form and the sign functions of $\{\omega_k\}$ appear in it. Combining our results with the RH problem and the Lax pair presented in \cite{WuXu20}, we develop direct relationships between solutions of the PDEs and the coupled $P_{IV}$ system derived in \cite{WuXu20}. Such connections were established in \cite{LyuChen20} for the special case $m=2$ of our problem via a different method. The PDE and other identities obtained from the ladder operators were used therein instead of the RH problem and the Lax pair. The research strategy of the present paper was adopted in \cite{LyuChenXu22} to study the Hankel determinant generated by the Laguerre weight with $m$ jump discontinuities and a coupled Painlev\'{e} V system was found.

This paper is built up as follows. In the next section, we present the derivation of the ladder operators satisfied by the monic orthogonal polynomials $P_n(z;\vec{t}\,)$ and the three compatibility conditions $(S_1), (S_2),(S_2')$. Using these identities, we establish expressions for the recurrence coefficients and the coefficient of $z^{n-1}$ in the monic orthogonal polynomial $P_n(z;\vec{t}\,)$ in terms of $2m$ auxiliary quantities $\{R_{n,k},r_{n,k}, k=1,\cdots,m\}$ which are shown to satisfy a system of difference equations that can be iterated in $n$. Section 3 is devoted to the $t$-evolution and differential relations are derived. Combining these relations with the results in Section 2, we establish Toda equations for the recurrence coefficients and Riccati equations for $\{R_{n,k},r_{n,k}\}$ from which we derive for $\{R_{n,k}\}$ second order PDEs that are reduced to a $P_{IV}$ for $m=1$. Based on the results in Sections 2 and 3, we show in Section 4 that the logarithmic derivative of the Hankel determinant $D_n(\vec{t}\,)$ satisfy a second order PDE which is reduced to the $\sigma$-form of a $P_{IV}$ when $m=1$. Assuming that $t_k-t_1 (k=2,\cdots,m)$ are fixed and by considering the RH problem for $P_n(z;\vec{t}\,)$, we construct in the last section direct relationships between $\{R_{n,k},r_{n,k}\}$ and solutions of the coupled $P_{IV}$ system produced in \cite{WuXu20}.

\section{Ladder operator approach and difference equations}
In this section, we will describe the ladder operator approach and apply it to the Hankel determinant $D_n(\vec{t}\,)$ given by \eqref{def-Dn}.

Denote by $\{P_n(z;\vec{t}\,)\}$ the monic polynomials orthogonal with respect to
\begin{align}\label{weight-g}
w(z;\vec{t}\,)=w_0(z)\biggl(\omega_0+\sum\limits_{k=1}^m \omega_k\theta(z-t_k)\biggr),\qquad z\in[a,b],
\end{align}
where $w_0(z)$ is an arbitrary positive smooth function on $[a,b]$ with moments of all orders existing, and $w_0(a)=w_0(b)=0$. Here $\vec{t}, \omega_k (k=0,1,\cdots,m)$ and $\theta(\cdot)$ have the same meaning as in \eqref{weight}. For our problem, $w_0(z)={\rm e}^{-z^2}, z\in(-\infty,+\infty)$. Before focusing on this case, we will at first derive the lowering and raising operators for $P_n(z;\vec{t}\,)$ associated with  \eqref{weight-g} and three supplementary conditions.

\subsection{Ladder operators and three compatibility conditions}
To achieve our goal, we first present properties for $P_n(z;\vec{t}\,)$ which are indeed valid for generic monic orthogonal polynomials whose associated weight function has finite moments of all orders. By using \eqref{m-or}, one can get the following three-term recurrence relation
\begin{align}\label{OP-recu}
zP_n(z;\vec{t}\,)=P_{n+1}(z;\vec{t}\,)+\alpha_n(\vec{t}\,)P_n(z;\vec{t}\,)+\beta_n(\vec{t}\,)P_{n-1}(z;\vec{t}\,),\qquad n\geq0,
\end{align}
subject to the initial conditions
\[P_0(z;\vec{t}\,):=1,\qquad\qquad\beta_0(\vec{t}\,):=0,\qquad\qquad P_{-1}(z;\vec{t}\,):=0.\]
Inserting \eqref{defPn} into \eqref{OP-recu} gives us
\begin{align}\label{m-alp}
\alpha_n(\vec{t}\,)=&p(n,\vec{t}\,)-p(n+1,\vec{t}\,),\qquad n\geq0,
\end{align}
with $p(0,\vec{t}\,):=0$, and a telescopic sum of \eqref{m-alp} yields
\begin{align}\label{sumal}
\sum_{j=0}^{n-1}\alpha_j(\vec{t}\,)=-p(n,\vec{t}\,).
\end{align}
 Multiplying both sides of \eqref{OP-recu} by $P_{n-1}(z;\vec{t}\,)w(z;\vec{t}\,)$ and integrating from $-\infty$ to $+\infty$ with respect to $z$, one finds
\begin{align}\label{bth} \beta_n(\vec{t}\,)=&\frac{h_n(\vec{t}\,)}{h_{n-1}(\vec{t}\,)},\qquad n\geq1.
\end{align}
Combining it with \eqref{Dprodh} leads us to the relationship between $\beta_n(\vec{t}\,)$ and the Hankel determinant:
\begin{align}\label{btD}
\beta_n(\vec{t}\,)=\frac{D_{n+1}(\vec{t}\,)D_{n-1}(\vec{t}\,)}{D_n^2(\vec{t}\,)}.
\end{align}
Moreover, from \eqref{OP-recu} we find the following Christoffel-Darbox formula \cite{Ismail,Szego}
\begin{align}\label{CD}
\sum_{k=0}^{n-1}\frac{P_k(x;\vec{t}\,\,)P_k(y;\vec{t}\,)}{h_k(\vec{t}\,)}=\frac{P_n(x;\vec{t}\,)P_{n-1}(y;\vec{t}\,)-P_{n-1}(x;\vec{t}\,)P_n(y;\vec{t}\,)}{h_{n-1}(\vec{t}\,)(x-y)}.
\end{align}

By making use of this formula, the recurrence relation \eqref{OP-recu} and the definition of the monic orthogonal polynomials given by \eqref{m-or}-\eqref{defPn}, one derives the following lowering operator.

\begin{lemma}\label{lad-oper}
The monic polynomials $P_n(z;\vec{t}\,)$ orthogonal with respect to \eqref{weight-g} satisfy the following lowering operator
\begin{align}\label{lp}
\frac{d}{dz}P_n(z;\vec{t}\,)=\beta_n(\vec{t}\,)A_n(z;\vec{t}\,)P_{n-1}(z;\vec{t}\,)-B_n(z;\vec{t}\,)P_n(z;\vec{t}\,),
\end{align}
where $A_n(z;\vec{t}\,)$ and $B_n(z;\vec{t}\,)$ are given by
\begin{subequations}\label{defAnBn}
\begin{align}
A_n(z;\vec{t}\,):=&\frac{1}{h_n(\vec{t}\,)}\int_{a}^{b}\frac{{\rm v}_0'(z)-{\rm v}_0'(y)}{z-y}P_n^2(y;\vec{t}\,)w(y;\vec{t}\,)dy+\sum_{k=1}^m \frac{R_{n,k}(t_k;\vec{t}\,)}{z-t_k},\label{m-defAn}\\
B_n(z;\vec{t}\,):=&\frac{1}{h_{n-1}(\vec{t}\,)}\int_{a}^{b}\frac{{\rm v}_0'(z)-{\rm v}_0'(y)}{z-y}P_n(y;\vec{t}\,)P_{n-1}(y;\vec{t}\,)w(y;\vec{t}\,)dy+\sum_{k=1}^m\frac{r_{n,k}(t_k;\vec{t}\,)}{z-t_k},\label{m-defBn}
\end{align}
\end{subequations}
and ${\rm v}_0(z)=-\ln w_0(z)$. Here the auxiliary quantities $R_{n,k}(\vec{t}\,)$ and $r_{n,k}(\vec{t}\,)$ for $ k=1,\cdots,m$ are defined by
\begin{subequations}\label{Rnrn}
\begin{align}
R_{n,k}(\vec{t}\,):=&\omega_k\frac{w_0(t_k)}{h_n(\vec{t}\,)}P_n^2(t_k;\vec{t}\,),\label{defR}\\
r_{n,k}(\vec{t}\,):=&\omega_k\frac{w_0(t_k)}{h_{n-1}(\vec{t}\,)}P_n(t_k;\vec{t}\,)P_{n-1}(t_k;\vec{t}\,).\label{defr}
\end{align}
\end{subequations}
\end{lemma}

\begin{remark}
The above lemma was given in Remark 2 of \cite{BasorChen09}, and it can be derived via an argument similar to the one for the $m=1$ case presented in Lemma 1 of \cite{BasorChen09}.

The Laguerre case where $w_0(z)=z^{\alpha}{\rm e}^{-z}$ in \eqref{weight-g} with $z\in[0,\infty),\alpha>-1$ was studied in \cite{LyuChenXu22} and the lowering operator was presented in Lemma 2.1 therein for $\alpha>0$ which agrees with our results. However, when $-1<\alpha<0$, since $w_0(0)\neq0$, $A_n(z)$ and $B_n(z)$ have to be modified (see Remark 2.2 of \cite{LyuChenXu22}).
 \end{remark}

Using the lowering operator \eqref{lp}, the recurrence relation \eqref{OP-recu}, the Christoffel-Darbox formula \eqref{CD} and the definitions of $A_n(z;\vec{t}\,)$ and $B_n(z;\vec{t}\,)$ given by \eqref{defAnBn}, we derive the raising operator for $P_n(z;\vec{t}\,)$ and three compatibility conditions satisfied by $A_n(z;\vec{t}\,)$ and $B_n(z;\vec{t}\,)$, numbered $(S_1),(S_2),(S_2')$. These will be presented in the following lemma. Here is the outline of the derivation: using the recurrence relation and the definitions of $A_n$ and $B_n$, we obtain $(S_1)$; combining it with the lowering operator and the recurrence relation, we deduce the raising operator; with it, the lowering operator, $(S_1)$ and the recurrence relation, we come to $(S_2)$; a combination of $(S_2)$ and $(S_1)$ leads us to $(S_2')$. For ease of notations, unless necessary we will not display the dependence on $\vec{t}\,$ in relevant quantities in the rest of this paper.
\begin{lemma}
For the monic polynomials orthogonal with respect to \eqref{weight-g}, the following raising operator hold:
\begin{align}\label{rp}
P_{n-1}'(z)=\left(B_n(z)+{\rm v}_0'(z)\right)P_{n-1}(z)-A_{n-1}(z)P_n(z),
\end{align}
where ${\rm v}_0(z):=\ln w_0(z)$. Moreover, $A_n(z)$ and $B_n(z)$ defined by \eqref{defAnBn} satisfy the following three identities
\begin{align}
B_{n+1}(z)+B_n(z)=&\left(z-\alpha_n\right)A_n(z)-{\rm v}_0'(z),\tag{$S_1$}\\
1+\left(z-\alpha_n\right)\left(B_{n+1}(z)-B_n(z)\right)=&\beta_{n+1}A_{n+1}(z)-\beta_nA_{n-1}(z),\tag{$S_2$}\\
B_n^2(z)+{\rm v}_0'(z)B_n(z)+\sum_{j=0}^{n-1}A_j(z)=&\beta_nA_n(z)A_{n-1}(z).\tag{$S_2'$}
\end{align}
\end{lemma}
\begin{proof}
From the three-term recurrence relation \eqref{OP-recu} and the identity \eqref{bth}, we have
\begin{align}\label{rec-1}
\frac{(y-\alpha_n)P_n(y)}{h_n}=\frac{P_{n+1}(y)}{h_n}+\frac{P_{n-1}(y)}{h_{n-1}}.
\end{align}
With this equality and the definitions of $R_{n,k}$ and $r_{n,k}$ given by \eqref{Rnrn}, we get
\begin{align*}
r_{n+1,k}+r_{n,k}=(t_k-\alpha_n)R_{n,k}, \qquad\qquad k=1,\cdots,m.
\end{align*}
Combing it and \eqref{rec-1} with the definitions of $A_n(z)$ and $B_n(z)$ given by \eqref{defAnBn}, we obtain $(S_1)$.

Replacing $n$ by $n-1$ in the lowering operator \eqref{lp} and the recurrence relation \eqref{OP-recu}, and combining the resulting equations to eliminate $\beta_{n-1}P_{n-2}(z)$, we find
\[P_{n-1}'(z)=\left((z-\alpha_{n-1})A_{n-1}-B_{n-1}\right)P_{n-1}(z)-A_{n-1}P_n(z).\]
Using $(S_1)$ with $n-1$ in place of $n$ to simplify the coefficient of $P_{n-1}(z)$, we are led to the raising operator \eqref{rp}.

Now we make use of $(S_1)$, the ladder operators \eqref{lp} and \eqref{rp}, and the recurrence relation \eqref{rec-1} to deduce $(S_2)$. See below for the derivation. Differentiating both sides of the recurrence relation \eqref{OP-recu} with respect to $z$, we have
\[P_{n+1}'(z)+\left(\alpha_n-z\right)P_n'(z)+\beta_nP_{n-1}'(z)-P_n(z)=0.\]
Replacing $P_{n+1}'(z)$ and $P_n'(z)$ by using the lowering operator \eqref{lp} and $P_{n-1}'(z)$ by using the raising operator \eqref{rp}, we obtain
\begin{align*}
&\left(\beta_{n+1}A_{n+1}-\beta_nA_{n-1}+B_n(z-\alpha_n)-1\right)P_n(z)\\
&\qquad\qquad\qquad-B_{n+1}P_{n+1}(z)+\beta_n\left[(\alpha_n-z)A_n+B_n+{\rm v}_0'\right]P_{n-1}(z)=0.
\end{align*}
Getting rid of $P_{n+1}(z)$ by using the three-term recurrence relation \eqref{OP-recu}, in view of $(S_1)$, we find
\[\left[\beta_{n+1}A_{n+1}-\beta_nA_{n-1}+(B_n-B_{n+1})(z-\alpha_n)-1\right]P_n(z)=0,
\]
which gives us $(S_2)$.

Multiplying both sides of $(S_2)$ by $A_n(z)$ and eliminating $(z-\alpha_n)A_n(z)$ in the resulting equation by using $(S_1)$, we get
\begin{equation}\label{S1-S2}
\begin{aligned}
A_n(z)+\left[B_{n+1}^2(z)-B_n^2(z)\right]+{\rm v}_0'(z)\left[B_{n+1}(z)-B_n(z)\right]\\
=\beta_{n+1}A_{n+1}(z)A_n(z)-\beta_nA_n(z)A_{n-1}(z).
\end{aligned}
\end{equation}
Observe that the left hand side of \eqref{S1-S2} except for the first term, and the right hand side are both first order differences in $n$. Hence, we replace $n$ by $j$ in \eqref{S1-S2} and take a summation of it from $j=0$ to $j=n-1$, with the initial conditions $B_0(z)=A_{-1}(z)=0$, we arrive at $(S_2')$.
\end{proof}

\begin{remark}
The derivation of $(S_2)$ and $(S_2')$ presented in the above lemma can be found in \cite[Theorem 5.2]{Assche} where the potential function ${\rm v}(x)=-\ln w(x)$ was supposed to be differentiable and ${\rm v}'(x)$ was Lipschitz continuous. The Riemann-Hilbert problem for the orthogonal polynomials was used therein to deduce the ladder operators and $(S_1)$ followed from them.  This derivation was different from ours. See also \cite{Magnus}.

For polynomials orthonormal with respect to $w(x)$, with the potential function twice continuously differentiable, the ladder operators were described in \cite{ChenIsmail}.

Refer to \cite{ChenP}, \cite{BasorChen09} and \cite{ChenZhang} for the discussion of the ladder operators and three supplementary conditions for the Gaussian, Laguerre and Jacobi weight functions with one jump discontinuity respectively.
\end{remark}

\subsection{Difference equations}
We now focus on our weight function \eqref{weight}, namely,
$
w_0(z)={\rm e}^{-z^2}, a=-\infty,b=+\infty$ in \eqref{weight-g}.
Using \eqref{defAnBn} to calculate $A_n(z)$ and $B_n(z)$, and substituting them into $(S_1)$ and $(S_2')$, we obtain a series of difference equations. From these we establish expressions for the recurrence coefficients and the coefficient of $z^{n-1}$ in the monic orthogonal polynomial $P_n(z)$ in terms of the auxiliary quantities $\{R_{n,k},r_{n,k}\}$ which are shown to satisfy a system of difference equations that can be iterated in $n$.

 Substituting ${\rm v}_0(z)=-\ln w_0(z)=z^2$ into \eqref{defAnBn}, with the aid of the orthogonality relation \eqref{m-or}, we have the following result.

\begin{lemma} $A_n(z)$ and $B_n(z)$ defined in \eqref{defAnBn} with ${\rm v}_0(z)=z^2,a=-\infty,b=+\infty$ are given by
\begin{subequations}\label{AnBn}
\begin{align}
A_n(z)=&2+\sum_{k=1}^m\frac{R_{n,k}(\vec{t}\,)}{z-t_k},\label{An}\\
B_n(z)=&\sum_{k=1}^m \frac{r_{n,k}(\vec{t}\,)}{z-t_k},\label{Bn}
\end{align}
\end{subequations}
where $\{R_{n,k},r_{n,k}, k=1,\cdots,m\}$ are defined by \eqref{Rnrn}.
\end{lemma}
Inserting \eqref{AnBn} into $(S_1)$,  by equating on both sides of the resulting expression the coefficients of $z^0$ and $(z-t_k)^{-1}$, we get
\begin{align}
&-2\alpha_n+\sum_{k=1}^m R_{n,k}=0,\label{s1-1}\\
&r_{n+1,k}+r_{n,k}=(t_k-\alpha_n)R_{n,k}, \qquad k=1,\cdots,m. \label{s1-2}
\end{align}

Substituting  \eqref{AnBn} into the left hand side of $(S_2')$, we have
\[l.h.s.=\biggl(\sum\limits_{k=1}^m \frac{r_{n,k}}{z-t_k}\biggr)^2+\sum_{k=1}^m\frac{1}{z-t_k}\biggl(\sum_{j=0}^{n-1}R_{j,k}+2t_kr_{n,k}\biggr)+\sum_{k=1}^m 2r_{n,k}+2n.\]
Expanding the first term, we find
\begin{align}
\biggl(\sum_{k=1}^m \frac{r_{n,k}}{z-t_k}\biggr)^2=\sum_{k=1}^m\frac{r_{n,k}^2}{(z-t_k)^2}+&\sum_{1\leq k<j\leq m} \frac{2r_{n,k}r_{n,j}}{(z-t_k)(z-t_j)}\nonumber\\
=\sum_{k=1}^m\frac{r_{n,k}^2}{(z-t_k)^2}+&\sum_{1\leq k<j\leq m} \frac{2r_{n,k}r_{n,j}}{t_k-t_j}\cdot\frac{1}{z-t_k}\nonumber\\
+&\sum_{1\leq k<j\leq m} \frac{2r_{n,k}r_{n,j}}{t_j-t_k}\cdot\frac{1}{z-t_j}\label{S2'}\\
=\sum_{k=1}^m\frac{r_{n,k}^2}{(z-t_k)^2}+&\sum_{k\neq j} \frac{2r_{n,k}r_{n,j}}{t_k-t_j}\cdot\frac{1}{z-t_k},\label{sum-split}
\end{align}
where the second equality is due to the fact that
\begin{align*}
\frac{1}{(z-t_k)(z-t_j)}=\frac{1}{t_k-t_j}\biggl(\frac{1}{z-t_k}-\frac{1}{z-t_j}\biggr),\
\end{align*}
and the third identity is obtained by exchanging $j$ and $k$ in the third term of \eqref{S2'}. The left hand side of $(S_2')$ now reads
\begin{equation}\label{lhs}
\begin{aligned}
l.h.s.=&\sum_{k=1}^m\frac{r_{n,k}^2}{(z-t_k)^2}+\sum_{k=1}^m\frac{1}{z-t_k}\Biggl(\sum_{\substack{j=1 \\ j\neq k}}^m \frac{2r_{n,k}r_{n,j}}{t_k-t_j}+\sum_{j=0}^{n-1}R_{j,k}+2t_kr_{n,k}\Biggr)\\
&+\sum_{k=1}^m2r_{n,k}+2n.
\end{aligned}
\end{equation}
To continue, we insert \eqref{AnBn} into the right hand side of $(S_2')$ and get
\[r.h.s.=\beta_n\Biggl[\sum\limits_{k=1}^m\frac{R_{n,k}}{z-t_k}\cdot\sum\limits_{j=1}^m\frac{R_{n-1,j}}{z-t_j}+\sum_{k=1}^m\frac{2(R_{n-1,k}+R_{n,k})}{z-t_k}+4\Biggr].\]
The first term in the square bracket can be simplified to
\begin{align*}
\sum\limits_{k=1}^m\frac{R_{n,k}}{z-t_k}\cdot\sum\limits_{j=1}^m\frac{R_{n-1,j}}{z-t_j}
=&\sum\limits_{k=1}^m\frac{R_{n,k}R_{n-1,k}}{(z-t_k)^2}
+\sum\limits_{j\neq k}\frac{R_{n,k}R_{n-1,j}+R_{n,j}R_{n-1,k}}{(z-t_k)(t_k-t_j)},
\end{align*}
via an argument similar to the one used to derive \eqref{sum-split}. Hence the right hand side of $(S_2')$ reads
\begin{equation}\label{rhs}
\begin{aligned}
r.h.s.=&\sum_{k=1}^m\frac{\beta_nR_{n,k}R_{n-1,k}}{(z-t_k)^2}+4\beta_n\\
&+\sum_{k=1}^m\frac{\beta_n }{z-t_k}\Biggl[\sum_{\substack{j=1 \\ j\neq k}}^m\frac{R_{n,k}R_{n-1,j}+R_{n,j}R_{n-1,k}}{t_k-t_j}+2(R_{n-1,k}+R_{n,k})\Biggr].
\end{aligned}
\end{equation}
Comparing the coefficients of $(z-t_k)^{-2}, (z-t_k)^{-1}$ and $z^0$ in \eqref{lhs} and \eqref{rhs}, we get the following equations
\begin{align}\label{s2'-3}
&r_{n,k}^2=\beta_n R_{n,k}R_{n-1,k},
\end{align}
\begin{equation}\label{s2'-2}
\begin{aligned}
&\sum_{\substack{j=1 \\ j\neq k}}^m\frac{2 r_{n,k}r_{n,j}}{t_k-t_{j}}+\sum_{j=0}^{n-1} R_{j,k}+2t_kr_{n,k}\\
&\qquad=\beta_n\sum_{\substack{j=1 \\ j\neq k}}^m\frac{R_{n,k}R_{n-1,j}+R_{n-1,k}R_{n,j}}{t_k-t_{j}}+2\beta_n(R_{n,k}+R_{n-1,k}),\qquad
\end{aligned}
\end{equation}
for $k=1,\cdots,m$, and
\begin{align}\label{s2'-1}
&2\sum_{k=1}^m r_{n,k}+2n =4\beta_n.\qquad\qquad\qquad\qquad
\end{align}

We summarize \eqref{s1-1} and \eqref{s2'-1} in the following lemma.
\begin{lemma}\label{albtRr}
 The recurrence coefficients in \eqref{OP-recu} are expressed in terms of the auxiliary quantities by
 \begin{subequations}\label{alRbtr}
\begin{align}
\alpha_n=&\frac{1}{2}\sum_{k=1}^m R_{n,k},\label{alR}\\
\beta_n=&\frac{1}{2}\biggl(\sum_{k=1}^m r_{n,k}+n\biggr).\label{btr}
\end{align}
\end{subequations}
\end{lemma}

Using the above expressions and the difference equations \eqref{s1-2} and \eqref{s2'-3}, we establish the following system of difference equations for $\{R_{n,k},r_{n,k},k=1,\cdots,m\}$ that can be iterated in $n$.
\begin{lemma}
The auxiliary quantities $\{R_{n,k},r_{n,k},k=1,\cdots,m\}$ satisfy the following difference equations:
\begin{align*}
r_{n+1,k}=&-r_{n,k}+\biggl(t_k-\frac{1}{2}\sum_{j=1}^m R_{n,j}\biggr)R_{n,k},\qquad k=1,2,\cdots,m,\\
R_{n,1}=&\frac{2r_{n,1}^2}{\biggl(\sum\limits_{k=1}^m r_{n,k}+n\biggr)R_{n-1,1}},\\
R_{n,k}=&\frac{r_{n,k}^2R_{n-1,1}}{r_{n,1}^2R_{n-1,k}}R_{n,1},\qquad k=2,\cdots,m,
\end{align*}
which can be iterated in $n$ with the initial conditions
\[R_{0,k}=\frac{\omega_k{\rm e}^{-t_k^2}}{\int_{-\infty}^{+\infty}{\rm e}^{-x^2}\Big(\omega_0+\sum\limits_{k=1}^m \omega_k\cdot\theta(x-t_k)\Big)dx},\qquad r_{0,k}=0.\]
\end{lemma}
\begin{remark}
For $m=1$, the two auxiliary quantities $R_{n,1}$ and $r_{n,1}$ were shown in \cite{MinChen19} to satisfy second order nonlinear  difference equations which may be connected with the discrete Painlev\'{e} IV equation.
\end{remark}

Now we proceed to derive expressions for $p(n,\vec{t}\,)$, the coefficient of $z^{n-1}$ in the monic orthogonal polynomial $P_n(z;\vec{t}\,)$, in terms of the auxiliary quantities. We will see in the next section that these identities are crucial for the derivation of the PDE satisfied by the  logarithmic derivative of the Hankel determinant $D_n(\vec{t}\,)$.
\begin{lemma} We have the following two expressions for $p(n,\vec{t}\,)$ in terms of the auxiliary quantities
\begin{align}
p(n,\vec{t}\,)=&-\frac{1}{2}\sum_{j=0}^{n-1}\sum_{k=1}^m R_{j,k}\label{p-1}\\
=&-\frac{1}{2}\biggl(\sum_{k=1}^m r_{n,k}+n\biggr)\sum_{k=1}^m R_{n,k}-\sum_{k=1}^m\frac{r_{n,k}^2}{R_{n,k}}+\sum_{k=1}^m t_k r_{n,k}.\label{p-2}
\end{align}
\end{lemma}
\begin{proof}
Replacing $n$ by $j$ in \eqref{alR} and inserting it into \eqref{sumal}, we obtain \eqref{p-1}.

A summation of \eqref{s2'-2} from $k=1$ to $k=m$ gives us
\begin{equation}\label{sumR-1}
\begin{aligned}
\sum_{k=1}^m\sum_{j=0}^{n-1} R_{j,k}=&\beta_n\sum_{j\neq k}^m\frac{R_{n,k}R_{n-1,j}+R_{n-1,k}R_{n,j}-2 r_{n,k}r_{n,j}}{t_k-t_{j}}\\
&+2\beta_n\sum_{k=1}^m(R_{n,k}+R_{n-1,k})-2\sum_{k=1}^mt_kr_{n,k}.
\end{aligned}
\end{equation}
Note that the first summation term on the right hand side can be split into two parts $\sum_{j>k}$ and $\sum_{j<k}$. By the anti-symmetry of this term in $j$ and $k$, we know that it is zero. Removing $R_{n-1,k}$ from the second term of the right hand side of \eqref{sumR-1} by using \eqref{s2'-3}, in view of \eqref{btr} and \eqref{p-1}, we arrive at \eqref{p-2}.
\end{proof}

\section{Toda equations, Riccati equations and the generalized Painlev\'{e} IV equation}\label{TRPDE}
In this section, we establish differential relations by taking the derivative with respect to $t_k$ on both sides of the orthogonality relation \eqref{m-or} with $j=k=n$ and $j=n, k=n-1$. With these relations and the results in the previous section, we establish Toda equations for the recurrence coefficients and Riccati equations for the auxiliary quantities $\{R_{n,k}, r_{n,k}\}$.

Before concentrating on our weight function \eqref{weight}, we consider the more general case given by \eqref{weight-g}. We have the following differential relations.
\begin{lemma} For monic orthogonal polynomials $P_n(z;\vec{t}\,)$ associated with the weight function \eqref{weight-g}, the derivatives of its $L^2$-norm and the coefficient of $z^{n-1}$ in $P_n(z;\vec{t}\,)$ are directly related to the auxiliary quantities defined through \eqref{Rnrn} by
\begin{subequations}\label{Dhnp}
\begin{align}
\frac{\partial}{\partial t_k}\ln h_n(\vec{t}\,)=-R_{n,k}(\vec{t}\,),\label{Dhn}\\
\frac{\partial}{\partial t_k}p(n,\vec{t}\,)=r_{n,k}(\vec{t}\,),\label{Dp}
\end{align}
for $k=1,\cdots,m$.
\end{subequations}
Since $\beta_n=h_n/h_{n-1}$ \emph {(cf. \eqref{m-alp})} and $\alpha_n=p(n,\vec{t}\,)-p(n+1,\vec{t}\,)$ \emph{(cf. \eqref{bth})}, it follows that
\begin{align}
\frac{\partial}{\partial t_k}\ln \beta_n(\vec{t}\,)=&R_{n-1,k}(\vec{t}\,)-R_{n,k}(\vec{t}\,),\label{Dbt}\\
\frac{\partial}{\partial t_k}\alpha_n(\vec{t}\,)=&r_{n,k}(\vec{t}\,)-r_{n+1,k}(\vec{t}\,),\label{Dal}
\end{align}
for $k=1,\cdots,m$.
\end{lemma}
\begin{proof}
For the weight function \eqref{weight-g}, we have
\[\frac{\partial}{\partial t_k}w(x;\vec{t}\,)=-w_0(x)\sum\limits_{k=1}^m \omega_k\delta(x-t_k),\]
where we make use of the fact that $(d/dy)\theta(y)=\delta(y)$ with $\delta(\cdot)$ denoting the Dirac delta function.
 Differentiating both sides of
\[h_n(\vec{t}\,)=\int_{a}^{b}P_n^2(x;\vec{t}\,)w(x;\vec{t}\,)dx\]
with respect to $t_k$, we get
\[\frac{\partial}{\partial t_k}h_n(\vec{t}\,)=-h_n(\vec{t}\,)R_{n,k}(\vec{t}\,),\]
which gives us \eqref{Dhn}.
Taking the derivative of
\[0=\int_{a}^{b}P_n(x;\vec{t}\,)P_{n-1}(x;\vec{t}\,)w(x;\vec{t}\,)dx\]
with respect to $t_k$, we are led to \eqref{Dp}.
\end{proof}
\begin{remark}
The Hankel determinant associated with the Laguerre weight with jump discontinuities, namely $w_0(z)=z^{\alpha}{\rm e}^{-z}, z\in[0,+\infty)$ in \eqref{weight-g} was described in \cite{LyuChenXu22}. The results presented in Lemma 2.8 therein agree with the ones given in the above lemma.
\end{remark}

Summing over \eqref{Dbt} and \eqref{Dal} for $k=1,\cdots,m$, in view of the expressions \eqref{alR} and \eqref{btr}, we come to the following Toda equations for the recurrence coefficients.
\begin{lemma}\label{Toda-1}
We have
\begin{subequations}\label{Toda}
\begin{align}
\mathcal{L}\ln\beta_n=&2(\alpha_{n-1}-\alpha_n),\label{Toda-bt}\\
\mathcal{L} \alpha_n=&2(\beta_n-\beta_{n+1})+1,\label{Toda-al}
\end{align}
\end{subequations}
where $\mathcal{L}=\sum\limits_{k=1}^m \frac{\partial}{\partial t_k}$, from which follows
the second order differential-difference equation for $\beta_n$
\[\mathcal{L}^2\ln\beta_n=4\left(\beta_{n-1}-2\beta_n+\beta_{n+1}\right).\]
Here $\mathcal{L}^2=\frac{\partial^2}{\partial t_k^2}+2\sum\limits_{j<k}\frac{\partial^2}{\partial t_k \partial t_j}$.
\end{lemma}

\begin{remark}
For $m=1$ and $m=2$, our Toda equations \eqref{Toda} coincide with (2.32)-(2.33) and (3.21)-(3.22) in \cite{MinChen19} respectively.
\end{remark}

Now we proceed with the derivation of the Riccati equations satisfied by the auxiliary quantities $\{R_{n,k},r_{n,k}\}$. To achieve this, we combine the differential relations for the recurrence coefficients given by \eqref{Dbt} and \eqref{Dal} with the difference identities obtained in the previous section, and make use of the expressions presented in Lemma \ref{albtRr}.
\begin{lemma}
The auxiliary quantities $\{R_{n,k},r_{n,k},k=1,\cdots,m\}$ satisfy the following Riccati equations
\begin{subequations}\label{RicRr}
\begin{align}
\mathcal{L} R_{n,k}=&\biggl(\sum_{j=1}^mR_{n,j}-2t_k\biggr)R_{n,k}+4r_{n,k},\label{Ric-R}\\
\mathcal{L} r_{n,k}=&\frac{2r_{n,k}^2}{R_{n,k}}-\biggl(\sum_{j=1}^m r_{n,j}+n\biggr)R_{n,k},\label{Ric-r}
\end{align}
\end{subequations}
for $k=1,\cdots,m$, where $\mathcal{L}=\sum\limits_{k=1}^m \frac{\partial}{\partial t_k}$.
\end{lemma}
\begin{proof}
Noting that $\frac{\partial^2\ln h_n}{\partial t_k\partial t_j}=\frac{\partial^2\ln h_n}{\partial t_j\partial t_k}$, we find from \eqref{Dhn} that
\[\frac{\partial}{\partial t_j}R_{n,k}=\frac{\partial}{\partial t_k}R_{n,j},\]
which together with \eqref{alR} gives us
\begin{align}
\frac{\partial\alpha_n}{\partial t_k}=&\frac{1}{2}\mathcal{L} R_{n,k}.\label{Dal-1}
\end{align}
Similarly, by combining the fact that $\frac{\partial^2 p(n,\vec{t})}{\partial t_k\partial t_j}=\frac{\partial^2 p(n,\vec{t})}{\partial t_j\partial t_k}$ with \eqref{Dp} and \eqref{btr}, we get
\begin{align}\label{Dbt-1}
\frac{\partial\beta_n}{\partial t_k}=&\frac{1}{2}\mathcal{L} r_{n,k}.
\end{align}
Eliminating $r_{n+1,k}$ from \eqref{Dal} and \eqref{s1-2}, and inserting \eqref{Dal-1} and \eqref{alR} into the resulting equation, we come to the Riccati equation for $R_{n,k}$, i.e. \eqref{Ric-R}.

 To derive the Riccati equation for $r_{n,k}$, we rewrite \eqref{Dbt} as
\[\frac{\partial\beta_n}{\partial t_k}=\beta_nR_{n-1,k}-\beta_n R_{n,k}.\]
Removing $R_{n-1,k}$ by using \eqref{s2'-3}, and inserting \eqref{btr} into the resulting expression, in view of \eqref{Dbt-1}, we arrive at \eqref{Ric-r}.
\end{proof}

Solving $r_{n,k}$ from the first Riccati equation \eqref{Ric-R} and substituting it into the second one, we establish second order PDEs for $R_{n,k}$.
\begin{theorem}
The auxiliary quantities \{$R_{n,k}, k=1,\cdots,m$\} satisfy the following coupled second order PDEs
\begin{equation}\label{PDE-R}
\mathcal{L}^2R_{n,k}-\frac{(\mathcal{L} R_{n,k})^2}{2R_{n,k}}=R_{n,k}\Biggl[\frac{3}{2}\biggl(\sum_{j=1}^m R_{n,j}\biggr)^2-2\sum_{j=1}^m(t_k+t_j)R_{n,j}+2(t_k^2-2n-1)\Biggr],
\end{equation}
for $k=1,\cdots,m$. Here
\[\mathcal{L}=\sum\limits_{k=1}^m \frac{\partial}{\partial{t_k}},\qquad\qquad\mathcal{L}^2=\sum\limits_{k=1}^m \frac{\partial^2}{\partial t_k^2}+2\sum\limits_{j<k}\frac{\partial^2}{\partial t_k \partial t_j}.\]
\end{theorem}

\begin{remark}
When $m=1$, we have $k=1$. With $R_{n,1}(\vec{t}\,)$ replaced by $R_n(t_1)$, equations \eqref{PDE-R} are reduced to
\[R''_{n}(t_1)-\frac{( R'_{n}(t_1))^2}{2R_{n}(t_1)}=R_{n}(t_1)\biggl[\frac{3}{2} R_{n}^2(t_1)-2t_1R_{n}(t_1)+2(t_1^2-2n-1)\biggr],\]
which is accord with (27) of \cite{ChenP} and  (2.37) of \cite{MinChen19}. As was pointed out therein, the above equation can be transformed into a particular $P_{IV}$ satisfied by $R_n(-t_1)$. In this sense, we may treat \eqref{PDE-R} as an $m$-variable generalization of a $P_{IV}$.

When $m=2$, the weight function has two jump discontinuities and the associated Hankel determinant $D_n(t_1,t_2)$ was studied in \cite{MinChen19}. However, equations \eqref{PDE-R} and the Riccati equation \eqref{RicRr} were not displayed therein.
\end{remark}

\section{Generalized $\sigma$-form of a Painlev\'{e} IV equation}
In this section, we will make use of the results obtained in the previous sections to establish the PDE satisfied by the logarithmic derivative of the Hankel determinant
\begin{align*}
D_n\left(\,\vec{t}\,\right)=\det\left(\int_{-\infty}^{\infty} x^{i+j}w(x;\vec{t}\,)dx\right)_{i,j=0}^{n-1},
\end{align*}
where the weight function is given by \eqref{weight} reading
\[
w(x;\vec{t}\,)={\rm e}^{-x^2}\biggl(\omega_0+\sum_{k=1}^m \omega_k\cdot\theta(x-t_k)\biggr).\]

According to \eqref{Dprodh}, i.e.
\begin{align*}
D_n(\vec{t}\,)=\prod_{k=0}^{n-1}h_k(\vec{t}\,),
\end{align*}
in view of \eqref{Dhn} and \eqref{p-1}, we find
\begin{align}
\sigma_n(\vec{t}\,):=&\sum\limits_{k=1}^m \frac{\partial}{\partial{t_k}}\ln D_n(\vec{t}\,)\label{defs}\\
=&-\sum_{k=1}^m\sum_{j=0}^{n-1}R_{j,k}(\vec{t}\,)\nonumber\\
=&2p(n,\vec{t}\,).\label{sig-p}
\end{align}
These relations motivate us to study $\sigma_n(\vec{t}\,)$ by using the equations involving $p(n,\vec{t}\,)$ which are established in the previous sections. We find that $\sigma_n$ and the auxiliary quantities can express each other.
\begin{lemma} $\sigma_n(\vec{t}\,)$ is expressed in terms of the auxiliary quantities by
\begin{align}\label{sig-rR}
\sigma_n=-\biggl(\sum_{k=1}^m r_{n,k}+n\biggr)\sum_{k=1}^m R_{n,k}-2\sum_{k=1}^m\frac{r_{n,k}^2}{R_{n,k}}+2\sum_{k=1}^m t_kr_{n,k}.
\end{align}
In turn, the auxiliary quantities are expressed in terms of $\sigma_n$ and its derivatives by
\begin{subequations}\label{rR-sigma}
\begin{gather}\label{r-sig}
r_{n,k}=\frac{1}{2}\cdot\frac{\partial\sigma_n}{\partial t_k},
\end{gather}
\begin{gather}\label{R-sig}
R_{n,k}=\frac{1}{2(\mathcal{L}\sigma_n+2n)}\Biggl[-\biggl(\frac{\partial^2\sigma_n}{\partial{t_k^2}}+\sum\limits_{\substack{j=1 \\ j\neq k}}^m\frac{\partial^2\sigma_n}{\partial{t_k}\partial{t_j}}\biggr)+{\rm sgn}(\omega_k)\sqrt{\Delta_k(\vec{t}\,)}\Biggr],
\end{gather}
\end{subequations}
for $k=1,\cdots,m$, where $\mathcal{L}=\sum\limits_{k=1}^m \frac{\partial}{\partial{t_k}}$ and ${\rm sgn}(\omega_k)$ is the sign function of $\omega_k$ which is $-1$ for negative $\omega_k$, $1$ for positive $\omega_k$ and $0$ for $\omega_k=0$. Here $\Delta_k(\vec{t}\,)$ is defined by
\begin{align}\label{def-Delta}
\Delta_k(\vec{t}\,):=\biggl(\frac{\partial^2\sigma_n}{\partial{t_k^2}}+\sum\limits_{\substack{j=1 \\ j\neq k}}^m\frac{\partial^2\sigma_n}{\partial{t_k}\partial{t_j}}\biggr)^2+4\left(\mathcal{L}\sigma_n+2n\right)\left(\frac{\partial\sigma_n}{\partial t_k}\right)^2.
\end{align}
\end{lemma}
\begin{proof}
Combining \eqref{sig-p} with \eqref{p-2} and \eqref{Dp} gives us \eqref{sig-rR} and \eqref{r-sig} respectively.

To continue, we insert \eqref{r-sig} into \eqref{btr} and get
\begin{align}\label{bt-Dsig}
\beta_n=\frac{1}{4}\mathcal{L}\sigma_n+\frac{n}{2}.
\end{align}
Getting rid of $R_{n-1,k}$ from \eqref{Dbt} by using \eqref{s2'-3}, we find
\begin{align*}
\beta_nR_{n,k}^2+\frac{\partial\beta_n}{\partial t_k}\cdot R_{n,k}-r_{n,k}^2=0,
\end{align*}
for $k=1,\cdots,m$.
Solving $R_{n,k}$ from this equation, in view of \eqref{r-sig} and \eqref{bt-Dsig}, we obtain
\begin{align}\label{Rsol-1}
R_{n,k}=&\frac{1}{2\beta_n}\left[-\frac{\partial\beta_n}{\partial t_k}\pm\frac{1}{4}\sqrt{\Delta_k(\vec{t}\,)}\right],
\end{align}
where $\Delta_k(\vec{t}\,)$ is defined by \eqref{def-Delta}.

Now we discuss the sign before the square root of \eqref{Rsol-1}. Replacing the term $\frac{\partial\beta_n}{\partial t_k}$ in \eqref{Rsol-1} by using \eqref{Dbt}, we get
\begin{align}\label{Rsol-2}
R_{n,k}+R_{n-1,k}=\pm\frac{1}{4\beta_n}\sqrt{\Delta_k(\vec{t}\,)}.
\end{align}
By the definition of $R_{j,k}$ given by \eqref{defR}, we know that
\[
{\rm sgn}\left(R_{n,k}+R_{n-1,k}\right)={\rm sgn} (\omega_k),\]
which indicates that the sign before the square root of \eqref{Rsol-2} and hence of \eqref{Rsol-1} is ${\rm sgn} (\omega_k)$, namely,
\begin{align*}
R_{n,k}=\frac{1}{2\beta_n}\left[-\frac{\partial\beta_n}{\partial t_k}+{\rm sgn}(\omega_k)\cdot\frac{1}{4}\sqrt{\Delta_k(\vec{t}\,)}\right].
\end{align*}
Inserting \eqref{bt-Dsig} into the above expression, we arrive at \eqref{R-sig}.
\end{proof}

Substituting \eqref{rR-sigma} back into \eqref{sig-rR}, we obtain the PDE satisfied by $\sigma_n$.
\begin{theorem}
$\sigma_n(\vec{t}\,)$ satisfies the following second order PDE
\begin{equation}\label{PDE-sig}
\sum_{k=1}^m {\rm sgn}(\omega_k)\sqrt{\Delta_k(\vec{t}\,)}=2\biggl(\sum_{k=1}^mt_k\frac{\partial\sigma_n}{\partial t_k}-\sigma_n\biggr),
\end{equation}
where $\Delta_k(\vec{t}\,)$ is given by \eqref{def-Delta}.
\end{theorem}
\begin{proof}
Using \eqref{s2'-3} to eliminate $R_{n-1,k}$ from \eqref{Dbt}, we get
\begin{align*}
\frac{r_{n,k}^2}{R_{n,k}}=\beta_nR_{n,k}+\frac{\partial\beta_n}{\partial t_k},\qquad k=1,\cdots,m.
\end{align*}
A summation of these equations, in view of \eqref{btr}, gives us
\[\sum_{k=1}^m\frac{r_{n,k}^2}{R_{n,k}}=\mathcal{L}\beta_n+\frac{1}{2}\biggl(\sum_{k=1}^m r_{n,k}+n\biggr)\sum_{k=1}^mR_{n,k}.\]
Inserting it into \eqref{sig-rR} yields
\begin{align*}
\sigma_n=-2\mathcal{L}\beta_n-2\biggl(\sum_{k=1}^m r_{n,k}+n\biggr)\sum_{k=1}^m R_{n,k}+2\sum_{k=1}^m t_kr_{n,k}.
\end{align*}
Plugging \eqref{rR-sigma} and \eqref{bt-Dsig} into this equation, we are led to \eqref{PDE-sig}.
\end{proof}
\begin{remark}
When $m=1$, the weight function \eqref{weight} has only one jump and the associated Hankel determinant $D_n(t_1)$ denotes the smallest eigenvalue distribution of GUE with $\omega_0=0, \omega_1=1$ and the largest eigenvalue distribution with $\omega_0=1,\omega_1=-1$, up to a constant. Equation \eqref{PDE-sig} is now reduced to, by clearing the square root,
\begin{equation*}
\left(\sigma_n''\right)^2+4\left(\sigma_n'+2n\right)\left(\sigma_n'\right)^2=4\left(t_1\sigma_n'-\sigma_n\right)^2.
\end{equation*}
It agrees with (5.14) of \cite{TW163} and (2.41) of \cite{MinChen19} which were identified to be the Jimbo-Miwa-Okamoto \cite{JimboMiwa1981} $\sigma$-form of a particular $P_{IV}$. In this sense, we may regard \eqref{PDE-sig} as an $m$-variable generalization of the $\sigma$-form of a $P_{IV}$.

When $m=2$, our weight function \eqref{weight} has jumps at $t_1$ and $t_2$. With the parameters $\omega_0=0,\omega_1=1,\omega_2=-1$ or $\omega_0=1,\omega_1=-1,\omega_2=1$, the associated Hankel determinant $D_n(t_1,t_2)$ represents the probability that all the eigenvalues of GUE lie in or outside the interval $(t_1,t_2)$ respectively, up to a constant. Equation \eqref{PDE-sig} is now reduced to
\[\biggl[4\biggl(t_1\frac{\partial\sigma_n}{\partial t_1}+t_2\frac{\partial\sigma_n}{\partial t_2}-\sigma_n\biggr)^2-\Delta_1-\Delta_2\biggr]^2=4\Delta_1\Delta_2,\]
where $\Delta_1$ and $\Delta_2$ are given by \eqref{def-Delta} with $k=1,2$. This equation coincides with (3.31) of \cite{MinChen19}.
\end{remark}

In the next section, we assume $t_j=c_j+t_1$ for $j=2,\cdots,m$. Given $c_2<\cdots<c_m$, our weight function \eqref{weight} and the associated Hankel determinant \eqref{def-Dn} have only one variable $t_1$. We will make use of the aforementioned results to show that the Hankel determinant is the $\tau$-function of a $P_{IV}$ and it admits an integral representation in terms of the auxiliary quantities $\{R_{n,k}\}$ which satisfy coupled ODEs in $t_1$. Moreover, by considering the RH problem satisfied by the monic orthogonal polynomials, we build relationships between the auxiliary quantities $\{R_{n,k},r_{n,k}\}$ introduced in the ladder operators and the solutions of a coupled $P_{IV}$ system.
\section{The coupled Painlev\'{e} IV system}
We introduce new variables
\[x:=t_1,\qquad\qquad c_j:=t_j-t_1, \quad j=2,\cdots,m\]
and $c_1:=0$. Since $t_1<\cdots<t_m$, we have $c_2<\cdots<c_m$. Write \[\vec{c}=(c_1,c_2,\cdots,c_m)=(0,t_2-t_1,\cdots,t_m-t_1).\] We have \[\vec{t}=(t_1,t_2,\cdots,t_m)=(x,c_2+x,\cdots,c_m+x).\]
In this section, we assume that $c_k, j=2,\cdots,m$ are constants. Hence the quantities which are functions of $\vec{t}$ now become functions of $x$.

\subsection{$\tau$-function of a Painlev\'{e} IV equation and coupled ODEs}
Denote the Hankel determinant $D_n(\vec{t}\,)$ by $D_n(x;\vec{c}\,)$ and define
\begin{align}\label{defDhat}
\hat{D}_n(x;\vec{c}):={\rm e}^{nx^2}D_n(x;\vec{c}).
\end{align}
According to Lemma \ref{albtRr} and the differential equations presented in Section \ref{TRPDE}, we show that $\hat{D}_n(x;\vec{c})$ is the $\tau$-function of a $P_{IV}$ and the auxiliary quantities $\{R_{n,k}\}$ defined by \eqref{defR} satisfy a system of second order ODEs in the variable $x$. Moreover, we establish the integral representation for $\hat{D}_n(x;\vec{c}\,)$ in terms of $\{R_{n,k}\}$ and their first order derivatives with respect to $x$.
\begin{theorem}\label{ThRode}
The quantity $\hat{D}_n(x;\vec{c})$ satisfies the following Toda molecule equation
\begin{align}\label{Toda-molecule}
\frac{d^2}{dx^2}\ln \hat{D}_n=4\frac{\hat{D}_{n+1}\hat{D}_{n-1}}{\hat{D}_n^2},
\end{align}
which indicates that $\hat{D}_n$ is the $\tau$-function of a $P_{IV}$. Moreover, $\hat{D}_n(x;\vec{c}\,)$ has the following integral representation:
\begin{equation}\label{Dint}
\begin{aligned}
\frac{\hat{D}_n(x;\vec{c}\,)}{\hat{D}_n(0;\vec{c})}=&{\rm e}^{nx^2}\exp\biggl[\frac{1}{4}\int_0^{x}\biggl(-\sum_{k=1}^m\frac{(R_{n,k}')^2 }{2R_{n,k}}+\frac{1}{2}\biggl(\sum_{k=1}^m R_{n,k}\biggr)^3+\sum_{k=1}^m 2(c_k+x)^2 R_{n,k}\biggr.\\
\biggl.&\qquad\qquad\quad\qquad-\biggl(\sum_{k=1}^m 2(c_k+x)R_{n,k}+4n\biggr)\cdot\sum_{j=1}^mR_{n,j}\biggr)dy\biggr],
\end{aligned}
\end{equation}
where $R_{n,k}=R_{n,k}(y;\vec{c}\,)=R_{n,k}(\vec{t}\,)|_{t_1=y,t_j=c_j+y,j=2,\cdots,m}$, $'$ denotes $d/dy$ and $\{R_{n,k}(\vec{t})=R_{n,k}(x;\vec{c})\}$ satisfy the following coupled ODEs:
\begin{equation}\label{Rode}
\begin{aligned}
&\frac{d^2}{dx^2}R_{n,k}-\frac{1}{2R_{n,k}}\biggl(\frac{d}{dx} R_{n,k}\biggr)^2\\
&=R_{n,k}\biggl[\frac{3}{2}\bigg(\sum_{j=1}^m R_{n,j}\bigg)^2-2\sum_{j=1}^m(c_k+c_j+2x)R_{n,j}+2((c_k+x)^2-2n-1)\biggr],
\end{aligned}
\end{equation}
for $k=1,\cdots,m$.
\end{theorem}
\begin{proof}
Since $x=t_1,c_1=0,c_j=t_j-t_1, j=2,\cdots,m$, we have $t_j=c_j+x$ for $j=1,\cdots,m$ and
\begin{align}\label{Ldx}
\mathcal{L}=\sum\limits_{k=1}^m \frac{\partial}{\partial{t_k}}=\frac{d}{dx}.
\end{align}
Replacing $\mathcal{L}$ by $d/dx$ in \eqref{PDE-R}, we get \eqref{Rode}.
From \eqref{defs} and \eqref{Ldx}, it follows that
\begin{align}\label{sdxD}
\sigma_n(\vec{t}\,)=\frac{d}{dx}\ln D_n(x;\vec{c}\,),
\end{align}
which combined with \eqref{bt-Dsig} gives us
\[\frac{d^2}{dx^2}\ln D_n(x;\vec{c}\,)=2(2\beta_n-n).\]
Plugging \eqref{btD} into the above equation, in view of \eqref{defDhat}, we come to \eqref{Toda-molecule}.

To derive \eqref{Dint}, we recall that $\sigma_n(\vec{t})$ is expressed in terms of the auxiliary quantities $\{R_{n,k},r_{n,k}\}$ by \eqref{sig-rR}. Solving $r_{n,k}$ from the Riccati equation \eqref{Ric-R} and substituting it into \eqref{sig-rR}, we find
\begin{equation}\label{x-sR}
\begin{aligned}
4\sigma_n=&-\sum_{k=1}^m\frac{(R_{n,k}')^2 }{2R_{n,k}}+\frac{1}{2}\biggl(\sum_{k=1}^m R_{n,k}\biggr)^3+\sum_{k=1}^m 2(c_k+x)^2 R_{n,k}\\
&-\biggl(\sum_{k=1}^m 2(c_k+x)R_{n,k}+4n\biggr)\cdot\sum_{j=1}^mR_{n,j},
\end{aligned}
\end{equation}
where $'$ stands for $d/dx$.
According to \eqref{sdxD} and \eqref{defDhat}, we integrate both sides of \eqref{x-sR} and obtain \eqref{Dint}.
\end{proof}

The next goal of this section is to build relationships between the auxiliary quantities $\{R_{n,k}(\vec{t}\,),r_{n,k}(\vec{t}\,)\}$ and solutions of a coupled $P_{IV}$ system. To achieve this, we make use of the RH problem and the Lax pair satisfied by the monic orthogonal polynomials associated with the weight function \eqref{weight}, which were described in \cite{WuXu20}. Hence, for convenience of our discussions, we first restate the relevant results of \cite{WuXu20} in the next subsection.

\subsection{RH problem, Lax pair and the coupled Painlev\'{e} IV system}
In \cite{WuXu20}, the following weight function with $m$ jump discontinuities was studied
\begin{align}
w(x;\vec{t}\,,\vec{\omega})=&{\rm e}^{-x^2}
\begin{cases}
1, &x\leq t_1,\\
\hat{\omega}_k, &t_k<x<t_{k+1},\;1\leq k\leq m-1,\\
\hat{\omega}_m, &x\geq t_m,
\end{cases}\nonumber\\
=&{\rm e}^{-x^2}\biggl(1+\sum_{k=1}^m(\hat{\omega}_k-\hat{\omega}_{k-1})\theta(x-t_k)\biggr).\label{weight-1}
\end{align}
where $\vec{\omega}=(\hat{\omega}_0, \hat{\omega}_1,\cdots,\hat{\omega}_m)$ with $\hat{\omega}_0:=1$ and $\theta(\cdot)$ has the same meaning as in \eqref{weight}. It is a special case of our weight function given by \eqref{weight} with $\omega_0=1$ and $ \omega_k=\hat{\omega}_k-\hat{\omega}_{k-1}$.

Denote by $P_n(z;\vec{t}\,,\vec{\omega})$ the monic polynomials orthogonal with respect to $w(z;\vec{t}\,,\vec{\omega})$ and define
\begin{align}\label{Y}
Y(z;\vec{t}\,,\vec{\omega}):=
\begin{pmatrix}
P_n(z;\vec{t}\,,\vec{\omega})&\frac{1}{2\pi i}\int_\mathbb{R} \frac{P_n(x;\vec{t}\,,\vec{\omega})w(x;\vec{t}\,,\vec{\omega})}{x-z}dx\\
\frac{-2\pi i}{h_{n-1}} P_{n-1}(z;\vec{t}\,,\vec{\omega})&-\frac{1}{h_{n-1}}\int_\mathbb{R}\frac{P_{n-1}(x;\vec{t}\,,\vec{\omega})w(x;\vec{t}\,,\vec{\omega})}{x-z}dx
\end{pmatrix}.
\end{align}
It is shown in \cite[section 2.5]{WuXu20} that
\begin{align}\label{phi}
\Phi(z;x,\vec{c},\vec{\omega}):=\sigma_1{\rm e}^{\frac{x^2}{2}\sigma_3}Y(z+x;\vec{t}\,,\vec{\omega}){\rm e}^{-\frac{1}{2}(z+x)^2\sigma_3}\sigma_1,
\end{align}
with the Pauli matrices $\sigma_1$ and $\sigma_3$ given by
\[\sigma_1=
\begin{pmatrix}
0&1\\
1&0
\end{pmatrix},\qquad\qquad\sigma_3=\begin{pmatrix}
1&0\\
0&-1
\end{pmatrix},\]
satisfies the following RH problem:
\begin{enumerate}
\item $\Phi(z;x,\vec{c},\vec{\omega})$ is analytic in $\mathbb{C}\backslash\mathbb{R};$
\item $\Phi(z;x,\vec{c},\vec{\omega})$ satisfies the jump condition
    \[\Phi_{+}(z;x,\vec{c},\vec{\omega})=\Phi_{-}(z;x,\vec{c},\vec{\omega})\cdot M,
    \]
    where the jump matrix $M$ is given by
   \begin{align*}
   M=
    \begin{pmatrix}
    1&0\\
    1&1
    \end{pmatrix},
   \qquad\qquad
    \begin{pmatrix}
    1&0\\
    \hat{\omega}_{k-1}&1
    \end{pmatrix},
    \qquad\qquad
    \begin{pmatrix}
    1&0\\
    \hat{\omega}_m&1
    \end{pmatrix}&\end{align*}
for $z<0, ~c_{k-1}<z<c_k ~(2\leq k\leq m)$ and $z>c_m$ respectively;
\item As $z\rightarrow\infty$,
\begin{align*}
\Phi(z;x,\vec{c},\vec{\omega})=\biggl(I+\frac{\Phi_1}{z}+\frac{\Phi_2}{z^2}+O(z^{-3})\biggr){\rm e}^{(z^2/2+xz)\sigma_3}z^{-n\sigma_3};
\end{align*}
\item As $z\rightarrow c_k$,
\begin{equation}\label{phick}
\Phi(z;x,\vec{c},\vec{\omega})=\Phi^{(c_k)}(z;x,\vec{c},\vec{\omega})
\begin{pmatrix}
    1&0\\
    \frac{\hat{\omega}_{k-1}-\hat{\omega}_k}{2\pi i}\ln (z-c_k)&1
    \end{pmatrix}E^{(k)},
    \end{equation}
where ${\rm arg}(z-c_k)\in(-\pi,\pi)$ for $1\leq k\leq m$ and $\Phi^{(c_k)}(z;x,\vec{c},\vec{\omega})$ is analytic near $z=c_k$ with the expansion
\begin{equation}\label{phick-1}
\Phi^{(c_k)}(z;x,\vec{c},\vec{\omega})=Q_{k0}(x)\left(I+Q_{k1}(x)(z-c_k)+O((z-c_k)^2)\right),
\end{equation}
for $1\leq k\leq m$. Here $E^{(k)}=I$ for ${\rm Im} z>0$ and $E^{(k)}=\begin{pmatrix}
    1&0\\
    -\hat{\omega}_k&1
    \end{pmatrix}$ for ${\rm Im}z<0.$
\end{enumerate}

By studying the above RH problem, it was shown in \cite{WuXu20} that $\Phi(z;x,\vec{c},\vec{\omega})$ has the following Lax pair
\begin{align}
\Phi_z(z;x)=&A(z;x)\Phi(z;x),\label{LP-z}\\
\Phi_x(z;x)=&B(z;x)\Phi(z;x),\nonumber
\end{align}
where $A(z;x)$ and $B(z;x)$ are given by
\begin{align}
A(z;x)=&(z+x)\sigma_3+A_{\infty}(x)+\sum_{k=1}^m\frac{A_k(x)}{z-c_k},\label{LP-A}\\
B(z;x)=&z\sigma_3+A_{\infty}(x),\nonumber
\end{align}
and
\begin{align}
A_{\infty}(x)=&\begin{pmatrix}
    0&y(x)\\
    -2\biggl(\sum\limits_{k=1}^m a_k(x)b_k(x)+n\biggr)/y(x)&0
    \end{pmatrix},\nonumber\\
    A_k(x)=&\begin{pmatrix}
    a_k(x)b_k(x)&a_k(x)y(x)\\
    -a_k(x)b_k^2(x)/y(x)&-a_k(x)b_k(x)
    \end{pmatrix},\qquad 1\leq k\leq m.\label{defAk}
 \end{align}
Here $y, a_k$ and $b_k$ are some functions of $x$. By combining the compatibility condition $\Phi_{zx}(z;x)=\Phi_{xz}(z;x)$ with the Lax pair, $a_k(x)$ and $b_k(x)$ were shown in \cite{WuXu20} to satisfy a coupled $P_{IV}$ system whose Hamiltonian is $\sigma_n(\vec{t})$ defined by \eqref{defs}.

\begin{lemma}
\emph {(\cite[Proposition 6]{WuXu20})} The quantity \[\sigma_n(\vec{t}\,)=\sum\limits_{k=1}^m \frac{\partial}{\partial{t_k}}\ln D_n(\vec{t}\,)=\frac{d}{dx}D_n(x;\vec{c})\]
is expressed in terms of $a_k(x)$ and $b_k(x)$ by
\[\sigma_n(\vec{t}\,)=2\sum\limits_{k=1}^m a_kb_k(x+c_k)-2\biggl(\sum\limits_{k=1}^m a_kb_k+n\biggr)\sum\limits_{k=1}^m a_k-\sum\limits_{k=1}^m a_kb_k^2,\]
where $a_k$ and $b_k$ satisfy the following coupled $P_{IV}$ system
\begin{equation}\label{cPiv}
\begin{cases}
\frac{d a_k}{dx}=-2a_k\biggl(\sum\limits_{j=1}^m a_j+b_k-x-c_k\biggr),\\
\frac{d b_k}{dx}=b_k^2+2b_k\biggl(\sum\limits_{j=1}^m a_j-x-c_k\biggr)+2\biggl(\sum\limits_{j=1}^m a_jb_j+n\biggr),
\end{cases}
\end{equation}
which is equivalent to the following Hamiltonian equations
\[\frac{da_k}{dx}=\frac{\partial \sigma_n}{\partial b_k},\qquad\qquad \frac{db_k}{dx}=-\frac{\partial \sigma_n}{\partial a_k},\qquad\qquad 1\leq k\leq m. \]
\end{lemma}

When $m=2$ in the weight function \eqref{weight-1}, by studying the above-mentioned RH problem and with the Lax pair, the recurrence coefficient $\beta_n(\vec{t}\,)$ of the associated monic orthogonal polynomials can be expressed in terms of $a_k$ and $b_k$, and $y(x)$ that appears in \eqref{defAk} is intimately related to $h_{n-1}(\vec{t}\,)$. These results were presented in Theorem 1 of \cite{WuXu20} and can be generalized to generic $m$ naturally.

\begin{lemma}  The following relationships are valid:
\begin{align}
\beta_n(\vec{t}\,)=&\frac{1}{2}\biggl(\sum_{k=1}^m a_k(x)b_k(x)+n\biggr),\label{btab}\\
y(x)=&\frac{4\pi i {\rm e}^{-x^2} }{{h_{n-1}(\vec{t}\,)}},\label{yhe}
\end{align}
where $i$ is the imaginary unit.
\end{lemma}

In the next subsection, we will make use of the results presented above and the ones obtained via the ladder operator approach to derive the relationships between $\{R_{n,k},r_{n,k}\}$ and $\{u_k,v_k\}$.

\subsection{Auxiliary quantities and the coupled Painlev\'{e} IV system}
\begin{theorem}\label{ThRrab} Given $0=c_1<c_2<\cdots<c_m$, for $\vec{t}=(t_1,t_2,\cdots,t_m)=(x,c_1+x,\cdots,c_m+x)$, the auxiliary quantities $\{R_{n,k}(\vec{t}\,), r_{n,k}(\vec{t}\,), k=1,\cdots,m\}$ defined by \eqref{Rnrn} are connected with $\{a_k(x;\vec{c}),b_k(x;\vec{c}),k=1,\cdots,m\}$ which satisfy the coupled $P_{IV}$ system \eqref{cPiv} by
\begin{subequations}\label{Rrab}
\begin{align}
R_{n,k}(\vec{t}\,)=&\frac{a_kb_k^2}{\sum\limits_{j=1}^m a_jb_j+n},\label{Rab}\\
r_{n,k}(\vec{t}\,)=&a_kb_k,\label{rab}
\end{align}
\end{subequations}
for $k=1,\cdots,m$,
or equivalently,
\begin{align*}
a_k(x;\vec{c})=&\frac{r_{n,k}^2}{R_{n,k}\biggl(\sum\limits_{j=1}^m r_{n,j}+n\biggr)},\\
b_k(x;\vec{c})=&\frac{R_{n,k}}{r_{n,k}}\biggl(\sum\limits_{j=1}^m r_{n,j}+n\biggr).
\end{align*}
Moreover, $a_k$ is directly related to $R_{n-1,k}$ by
\begin{align}\label{akR}
a_{k}(x;\vec{c})=\frac{1}{2}R_{n-1,k},
\end{align}
for $k=1,\cdots,m$.
\end{theorem}
\begin{proof}
As $z\rightarrow c_k$, by substituting  \eqref{phick} into \eqref{LP-z}, we find
\begin{align*}
A(z;x)=&\Phi_z(z;x)\Phi^{-1}(z;x)\\
=&\Phi_z^{(c_k)}\left(\Phi^{(c_k)}\right)^{-1}+\Phi^{(c_k)}
\begin{pmatrix}
0&0\\
\frac{\hat{\omega}_{k-1}-\hat{\omega_k}}{2\pi i}\cdot\frac{1}{z-c_k}&0
\end{pmatrix}\left(\Phi^{(c_k)}\right)^{-1}.
\end{align*}
According to \eqref{phick-1} and \eqref{LP-A}, we obtain
\begin{align}
A_k=&\frac{\hat{\omega}_{k-1}-\hat{\omega_k}}{2\pi i}Q_{k0}\begin{pmatrix}
0&0\\
1&0
\end{pmatrix}
Q_{k0}^{-1}\nonumber\\
=&-\frac{\omega_k}{2\pi i}
\begin{pmatrix}
\left(Q_{k0}\right)_{12}\left(Q_{k0}\right)_{22}&-\left(Q_{k0}\right)_{12}^2\\
\left(Q_{k0}\right)_{22}^2&-\left(Q_{k0}\right)_{22}\left(Q_{k0}\right)_{12}
\end{pmatrix},\label{Ak-1}
\end{align}
where $\omega_k=\hat{\omega}_{k}-\hat{\omega}_{k-1}$.
From \eqref{phick-1}, we know that
\begin{align}\label{phiQ}
\Phi^{(c_k)}(c_k)=Q_{k0}(x).
\end{align}
Hence, to compute $A_k$, it suffices to work out $\Phi^{(c_k)}(c_k)$.
Combining \eqref{phick} with \eqref{phi}, we get
\begin{align}
\Phi^{(c_k)}(z)=&\Phi(z)\left(E^{(k)}\right)^{-1}\begin{pmatrix}
    1&0\\
    \frac{\omega_k}{2\pi i}\ln (z-c_k)&1
    \end{pmatrix}\nonumber\\
=&\begin{pmatrix}
{\rm e}^{-x^2/2-(z+x)^2/2}Y_{21}(z+x)& {\rm e}^{-x^2/2+(z+x)^2/2}Y_{22}(z+x)\\
{\rm e}^{x^2/2-(z+x)^2/2}Y_{11}(z+x)& {\rm e}^{x^2/2+(z+x)^2/2}Y_{12}(z+x)
\end{pmatrix}\cdot
\begin{pmatrix}
A&1\\
1&0
\end{pmatrix},\label{phick-2}
\end{align}
where $A=\hat{\omega}_k+\frac{\hat{\omega}_{k}-\hat{\omega}_{k-1}}{2\pi i}\ln (z-c_k)$ for ${\rm Im}z<0$ and $A=\frac{\hat{\omega}_{k}-\hat{\omega}_{k-1}}{2\pi i}\ln (z-c_k)$ for ${\rm Im}z>0$.

Now we look at $\left(\Phi^{(c_k)}(z)\right)_{12}$ and $\left(\Phi^{(c_k)}(z)\right)_{22}$. According to \eqref{phiQ} and \eqref{phick-2}, we find
\begin{align*}
(Q_{k0}(x))_{12}=&\left(\Phi^{(c_k)}(z)\right)_{12}|_{z=c_k}\\
=&{\rm e}^{-x^2/2-(z+x)^2/2}Y_{21}(z+x)|_{z=c_k}\\
=&\frac{-2\pi i}{h_{n-1}}{\rm e}^{-x^2/2-t_k^2/2} P_{n-1}(t_k;\vec{t}\,,\vec{\omega}),
\end{align*}
where the last equality is due to \eqref{Y}. Similarly, we can show that
\begin{align*}
(Q_{k0}(x))_{22}
=&{\rm e}^{x^2/2-t_k^2/2} P_{n}(t_k;\vec{t}\,,\vec{\omega}).
\end{align*}
Hence it follows from \eqref{Ak-1} that
\begin{align}
(A_k)_{11}=&-\frac{\omega_k}{2\pi i}\left(Q_{k0}\right)_{12}\left(Q_{k0}\right)_{22}\nonumber\\
=&\frac{\omega_k}{h_{n-1}}{\rm e}^{-t_k^2}P_n(t_k)P_{n-1}(t_k)=r_{n,k},\label{A11-1}
\end{align}
where the second equality is due to the definition of $r_{n,k}$ given by \eqref{defr}, and
\begin{align}\label{A21-1}
(A_k)_{21}=&-\frac{\omega_k}{2\pi i}\left(Q_{k0}\right)_{22}^2=-\frac{\omega_k}{2\pi i}{\rm e}^{x^2-t_k^2} P_{n}^2(t_k)=-\frac{h_n}{2\pi i}{\rm e}^{x^2}R_{n,k},
\end{align}
where the last equality results from \eqref{defR}. A combination of \eqref{A11-1} and \eqref{defAk} leads us to \eqref{rab}. Combining \eqref{A21-1} with \eqref{defAk}, in view of \eqref{btab}, \eqref{yhe} and the fact that $\beta_n=h_n/h_{n-1}$, we arrive at \eqref{Rab}.

Inserting \eqref{btab} and \eqref{Rrab} into \eqref{s2'-3}, we come to \eqref{akR}.
\end{proof}

\begin{remark}
Replacing $n$ by $n-1$ in \eqref{Rode} and inserting \eqref{akR} into the resulting expression, we are led to
\begin{equation*}
\begin{aligned}
&\frac{d^2a_k}{dx^2}-\frac{1}{2a_k}\biggl(\frac{da_k}{dx}\biggr)^2-4a_k\sum_{j=1}^ma_j\biggl(\sum_{i=1}^m a_i-x-c_j\biggr)\\
&\qquad\qquad-2a_k\biggl(\sum_{j=1}^m a_j-x-c_k\biggr)^2+2(2n-1)a_k=0,\quad 1\leq k\leq m.
\end{aligned}
\end{equation*}
These equations can also be obtained by eliminating $b_k$ from the coupled $P_{IV}$ system \eqref{cPiv}.
\end{remark}

\section*{Acknowledgments}
Yang Chen was supported by the Macau Science and Technology Development Fund under grant number FDCT 0079/2020/A2, by Guangdong Natural Science Foundation under grant number EF012/FST-CYY/2021/GDSTC, and by the University of Macau under grant number MYRG 2022-00014-FST. Shulin Lyu was supported by National Natural Science Foundation of China under grant numbers 12101343 and 11971492, and by Shandong Provincial Natural Science Foundation with project number ZR2021QA061.

}


\begin{thebibliography}{}

\bibitem{ACM}
M. Atkin, T. Claeys and F. Mezzadri, Random matrix ensembles with singularities and a hierarchy of Painlev\'{e} III equations, Int. Math. Res. Notices {\bf 2016} (2016), 2320--2375.

\bibitem{BasorChen09}
E. Basor and Y. Chen, Painlev\'{e} V and the distribuition function of a discontinuous linear statistic in the Laguerre unitary ensembels, J. Phys. A: Math. Theor. {\bf 42} (2009), 035203 (18pp).

\bibitem{BasorChenZhang}
E. Basor, Y. Chen and L. Zhang, PDEs satisfied by extreme eigenvalues distributions of GUE and LUE, Random Matrices-Theory Appl., {\bf 1} (2012), 1150003 (21pp).

\bibitem{BMM}
L. Brightmore, F. Mezzadri and M. Mo, A matrix model with a singular weight and Painlev\'{e} III, Commun. Math. Phys. {\bf 333} (2015), 1317--1364.

\bibitem{Charlier}
C. Charlier, Asymptotics of Hankel determinants with a one-cut regular potential and Fisher-Hartwig singularities, arXiv:1706.03579.

\bibitem{cd}
C. Charlier  and A.  Doeraene,  The generating function for the Bessel point process and a system of coupled Painlev\'e V equations,  Random Matrices-Theor. Appl.  {\bf  8} (2019), 1950008 (31pp).

\bibitem{ChenChenFan19}
M. Chen, Y. Chen and E. Fan, The Riemann-Hilbert analysis to the Pollaczek-Jacobi type orthogonal polynomials, Stud. Appl. Math. {\bf 143} (2019), 42--80.

\bibitem{ChenChenFan19-1}
M. Chen, Y. Chen and E. Fan, Critical edge behavior in the perturbed Laguerre unitary ensemble and the Painlev\'{e} V transcendent, J. Math. Anal. Appl. {\bf 474} (2019), 572--611.

\bibitem{ChenHaqMcKay13}
Y. Chen, N. Haq and M. McKay, Random matrix models, double-time Painlev\'{e} equations, and wireless relaying, J. Math. Phys. {\bf 54} (2013), 063506 (55pp).

\bibitem{ChenIsmail}
Y. Chen and M. Ismail, Ladder operators and differential equations for orthogonal polynomials, J. Phys. A: Math. Gen. {\bf 30} (1997), 7817--7829.

\bibitem{ChenIsmail04}
Y. Chen and M. Ismail, Jacobi polynomials from compatibility conditions, Proc. Amer. Math. Soc. {\bf 133} (2004), 465--472.

\bibitem{ChenIts10}
Y. Chen and A. Its, Painlev\'{e} III and a singular linear statistics in Hermitian random matrix ensembles, I, J. Approx. Theory {\bf 162} (2010), 270--297.

\bibitem{ChenMcKay}
Y. Chen and M. R. McKay, Coulumb fluid, Painlev\'{e} transcendents, and the information theory of MIMO system, IEEE Trans. Inf. Theory {\bf 58} (2012), 4594--4634.

\bibitem{ChenP}
Y. Chen and G. Pruessner, Orthogonal polynomials with discontinuous weights, J. Phys. A: Math. Gen. {\bf 38} (2005), L191- L198.

\bibitem{ChenZhang}
Y. Chen and L. Zhang, Painlev\'{e} VI and the unitary Jacobi ensembles, Stud. Appl. Math. {\bf 125} (2010), 91--112.

\bibitem{cd18}
T. Claeys and A. Doeraene, The generating function for the Airy point process
and a system of coupled Painlev\'{e} II equations, Stud. Appl. Math. {\bf 140} (2018), 403--437.


\bibitem{DaiXuZhang18}
D. Dai, S. Xu and L. Zhang, Gap probability at the hard edge for random matrix ensembles with pole singularities in the potential, SIAM J. Math. Anal. {\bf 50} (2018), 2233--2279.

\bibitem{DaiXuZhang19}
D. Dai, S. Xu and L. Zhang, Gaussian unitary ensembles with pole singularities near the soft edge and a system of coupled Painlev\'{e} XXXIV equations, Ann. Henri Poincar\'{e} {\bf 20} (2019), 3313--3364.


\bibitem{Deift}
P. Deift, Orthogonal Polynomials and Random Matrices: A Riemann-Hilbert Approach,  Courant Lecture Notes 3, New York University, Amer. Math. Soc., Providence, RI, 1999.

\bibitem{DeiftZhou}
P. Deift and X. Zhou, A steepest descent method for oscillatory Riemann–Hilbert problems, asymptotics for the MKdV equation, Ann. Math. {\bf 137} (1993), 295--368.

\bibitem{HanChen}
P. Han and Y. Chen, A degenerate Gaussian weight connected with Painlev\'{e} equations and Heun equations, Random Matrices-Theor. Appl. {\bf 10} (2021), 2150034 (32pp).

\bibitem{Ismail}
M. Ismail, Classical and Quantum Orthogonal Polynomials in One Variable, Encyclopedia of Mathematics and its Applications 98, Cambridge University Press, Cambridge, 2005.

\bibitem{JimboMiwa1981}
M. Jimbo and T. Miwa,
Monodromy perserving deformation of linear ordinary differential equations with rational coefficients. II,
Physica {\bf 2D} (1981), 407--448.

\bibitem{LyuChen20}
S. Lyu and Y. Chen, Gaussian unitary ensembles with two jump discontinuities, PDEs, and the coupled Painlev\'{e} II and IV systems, Stud. Appl. Math. {\bf 146} (2021), 118--138.

\bibitem{LyuChenXu22}
S. Lyu, Y. Chen and S. Xu, Laguerre unitary ensembles with jump discontinuities, PDEs and the coupled Painlev\'{e} V system, Physica D {\bf 449} (2023), 133755 (14pp).


\bibitem{LyuGriffinChen}
S. Lyu, J. Griffin and Y. Chen, The Hankel determinant associated with a singularly perturbed Laguerre unitary ensemble, J. Nonlinear Math. Phys. {\bf 26} (2019), 24--53.

\bibitem{Magnus}
A. P. Magnus, Painlev\'{e}-type differential equations for the recurrence coefficients of semi-classical orthogonal polynomials, J. Comput. Appl. Math. {\bf 57} (1995), 215--237.

\bibitem{Mehta}
M. Mehta, Random Matrices, 3rd edition, Elsevier, New York, 2004.

\bibitem{MinChen19}
C. Min and Y. Chen, Painlev\'{e} transcendents and the Hankel determinants generated by a discontinuous Gaussian weight, Math. Meth. Appl. Sci. {\bf 42} (2019), 301--321.

\bibitem{MinChen21}
C. Min and Y. Chen, Differential, difference, and asymptotic relations for Pollaczek-Jacobi type orthogonal polynomials and their Hankel determinants, Stud. Appl. Math. {\bf 147} (2021) 390--416.

\bibitem{MinChen22-AMP}
C. Min and Y. Chen, Semi-classical Jacobi polynomials, Hankel determinants and asymptotics, Anal. Math. Phys. {\bf 12} (2022), 8 (25pp).

\bibitem{Szego}
G. Szeg\"{o}, \textit{Orthogonal Polynomials}, American Mathematical Society Colloquium Publications, vol. 23, New York, 1939.


\bibitem{TW159}
C. Tracy and H. Widom, Level-spacing distributions and the Airy kernel, Commun. Math. Phys. {\bf 159} (1994), 151-174.

\bibitem{TW163}
C. Tracy and H. Widom, Fredholm determinants, differential equations and matrix models, Commun. Math. Phys. {\bf 163} (1994), 33--72.

\bibitem{Assche}
W. Van Assche, Orthogonal Polynomials and Painlev\'{e} Equations, Australian Mathematical
Society Lecture Series 27, Cambridge University Press, 2018.

\bibitem{WuXu20}
X. Wu and S. Xu, Gaussian unitary ensemble with jump discontinuities and the coupled Painlev\'{e} II and IV systems, Nonlinearity {\bf 34} (2021), 2070--2115.


\bibitem{XuZhao11}
S. Xu and Y. Zhao, Painlev\'{e} XXXIV asymptotics of orthogonal polynomials for the Gaussian weight with a jump at the edge, Stud. Appl. Math. {\bf 127} (2011), 67--105.

\bibitem{XuZhao20}
S. Xu and Y. Zhao, Gap probability of the circular unitary ensemble with a Fisher-Hartwig singularity and the coupled Painlev\'{e} V system, Commun. Math. Phys. {\bf 377} (2020), 1545--1596.

\end{thebibliography}
\end{document}